\documentclass[journal]{IEEEtran}

\usepackage{graphicx}
\usepackage{bm}
\usepackage{enumerate}
\usepackage{float}
\usepackage{multicol}
\usepackage{color}
\usepackage{url}
\usepackage{cite}

\usepackage{amsmath}
\usepackage{amsthm}
\usepackage{amssymb}
\DeclareMathOperator*{\argmax}{argmax}
\newtheorem{lem}{Lemma}
\usepackage{algorithmic}
\usepackage{algorithm}

\usepackage{array}

\ifCLASSOPTIONcompsoc
  \usepackage[caption=false,font=normalsize,labelfont=sf,textfont=sf]{subfig}
\else
  \usepackage[caption=false,font=footnotesize]{subfig}
\fi

\begin{document}
\captionsetup{belowskip=0pt,aboveskip=0pt}

\title{Rate Adaptation in Predictor Antenna Systems}
\author{Hao~Guo,
        Behrooz~Makki, ~\IEEEmembership{Senior~Member,~IEEE},
        and~Tommy~Svensson, ~\IEEEmembership{Senior~Member,~IEEE}

\thanks{This work was supported in part by VINNOVA (Swedish Government Agency for Innovation Systems) within the VINN Excellence Center ChaseOn, and in part by the EC within the H2020 project 5GCAR.}
\thanks{H. Guo and T. Svensson are with the Department of Electrical Engineering, Chalmers University of Technology, 41296 Gothenburg, Sweden (email: hao.guo@chalmers.se; tommy.svensson@chalmers.se).}
\thanks{B. Makki is with the Ericsson Research, 41756 Gothenburg, Sweden (email: behrooz.makki@ericsson.com).}}

\maketitle

\begin{abstract}
Predictor antenna (PA) system is referred to as a system with two sets of antennas on the roof of a vehicle, where the PAs positioned in the front of the vehicle are used to predict the channel state observed by the receive antennas (RAs) that are aligned behind the PAs. This letter studies the performance of PA systems in the presence of the mismatching problem, i.e., when the channel observed by the PA is not exactly the same as the channel experienced by the RA. Particularly, we study the effect of spatial mismatching on the accuracy of channel state information estimation and rate adaption. We derive closed-form expressions for instantaneous throughput, outage probability, and the throughput-optimized rate adaptation. Also, we take the temporal evolution of the channel into account and evaluate the system performance in temporally-correlated conditions. The simulation and analytical results show that, while PA-assisted adaptive rate adaptation leads to considerable performance improvement, the throughput and the outage probability are remarkably affected by the spatial mismatch and temporal correlations.
\end{abstract}


\vspace{-5mm}
\section{Introduction}
\vspace{-2mm}
Vehicle communication is one of the most important use cases in the fifth generation of wireless networks (5G). In 5G, a significant number of users would access wireless networks in vehicles, e.g., in public transportation like trams and trains or private cars, by their smart phones and laptops. Setting a moving relay node (MRN) in vehicles can be one promising solution to provide a high-rate reliable connection between a base station (BS) and the users inside the vehicle \cite{Yutao2013ICMmoving}. 

With an MRN, the channel state information at the transmitter (CSIT), i.e., at the BS, plays an important role on the system performance.  However, the typical CSIT acquisition methods such as Kalman filter-based system, which are mostly designed for static/low speed channels, may not work well for high-speed MRNs. This is because the position of the antennas change quickly and the channel information becomes inaccurate for MRNs. To overcome this issue, \cite{Sternad2012WCNCWusing} proposes the concept of predictor antenna (PA) wherein at least two antennas are deployed on top of the vehicle. The first antenna, which is the PA,  estimates the channel and sends feedback to the BS. Then, the BS uses the CSIT provided by the PA to communicate with a second antenna, which we refer to as receive antenna (RA), when it arrives to the same position as the PA.

In \cite{Sternad2012WCNCWusing,BJ2017ICCWusing}, experimental results are presented for the efficiency of the PA system which validate the feasibility of this concept. Furthermore,  \cite{phan2018WSAadaptive} proves the feasibility of the PA concept in massive multiple-input-multiple-output (MIMO) downlink systems. However, if the RA does not arrive in the same point as the PA, the actual channel for the RA would not be identical to the one experienced by the PA before. Different ways to compensate for this mismatching problem have been proposed in \cite{DT2015ITSMmaking}. On the other hand, the BS can potentially adjust the transmission rate according to the mismatching problem, if the spatial difference between the PA and the RA in successive time slots is known by the BS. However, to the best of our knowledge, there is no work on rate adaptation in PA systems.

In this paper, we study the problem of imperfect CSIT estimation in PA systems. The goal is to maximize the throughput in the presence of imperfect CSIT. We model the mismatching of the PA system as an equivalent non-central Chi-squared fading model and perform the rate adaptation maximizing the throughput. Particularly, we derive closed-form expressions for the throughput, outage probability as well as the throughput-optimized rate adaptation. Finally, we study the system performance in temporally-correlated fading conditions. Our results, which are of interest for different delay-constrained scenarios of vehicle communications, indicate that rate adaptation can effectively compensate for the mismatching problem in PA systems.

\vspace{-5mm}
\section{System Model}
\begin{figure}
\centering
  \includegraphics[width=0.78\columnwidth]{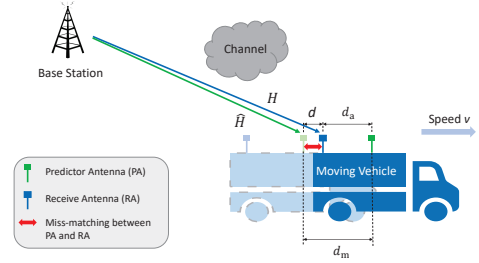}\\
\caption{Predictor antenna system with the mismatching problem.}\label{system}
\vspace{-5mm}
\end{figure}
As seen in Fig. \ref{system}, consider one vehicle deploying two antennas on the roof with one PA positioned in the front of the moving direction and an RA aligned behind the PA. The idea of the data transmission model is that the PA first sends pilots in time $t$, then the BS estimates the channel and sends the data in time $t+\delta$ to the RA. Here,  $\delta$ depends on the processing time at the BS. Here, we assume a time division duplex (TDD) system.

We initially assume that the vehicle moves through a stationary electromagnetic standing wave pattern. Thus, the vehicle experiences a time-invariant pattern. For our analysis of  the effect of temporal correlation, see Section III.A. Due to the movement, after some time the RA will come to the same place as the PA and observe the same channel. In this way, the BS can use the same CSIT which is known from the PA for data transmission to the RA. If the RA reaches exactly the same position as the position of the PA when sending the pilots, it will experience the same channel and the CSIT will be perfect. However, depending on the vehicle speed, the RA may receive the data in a position different from the one in which the PA was sending the pilots. In this case, the channel estimation is not the same as the actual channel in the BS-RA link, which will affect the system throughput correspondingly.

Considering downlink transmission in the BS-RA link, the received signal is given by
\begin{align}\label{eq_Y}
{{Y}} = \sqrt{P}HX + Z.
\end{align}
Here, $P$ represents the average received power at the RA, while $X$ is the input message with unit variance, and $H$ is the fading coefficient between the BS and the RA. Also, $Z \sim \mathcal{CN}(0,1)$ denotes the independent and identically distributed (IID) complex Gaussian noise added at the receiver.

We denote the channel coefficient of the PA-BS uplink as $\hat{H}$ and we assume that $\hat{H}$ is perfectly known by the BS. The result can be extended to the cases with imperfect CSI at the BS (see Section III.A). In this way,  we use the spatial correlation model \cite[p. 2642]{Shin2003TITcapacity}
\begin{align}\label{eq_tildeH}
    \Tilde{\bm{H}} = \bm{\Phi}^{1/2} \bm{H}_{\varepsilon},
\end{align}
where $\Tilde{\bm{H}}$ = $\bigl[ \begin{smallmatrix}
  \hat{H}\\H
\end{smallmatrix} \bigr]$ is the channel matrix including both BS-PA ($\hat{H}$) and BS-RA ($H$) links. $\bm{H}_{\varepsilon}$ has independent circularly-symmetric zero-mean complex Gaussian entries with unit variance, and $\bm{\Phi}$ is the channel correlation matrix.

In general, the spatial correlation of the fading channel depends on the distance between the RA and the PA, which we denote by $d_\text{a}$, as well as the angular spectrum of the radio wave pattern. If we use the classical Jakes' correlation model by assuming uniform angular spectrum, the $(i,j)$-th entry of $\bm{\Phi}$ is given by \cite[eq. (1)]{Chizhik2000CLeffect}
\begin{align}\label{eq_phi}
    \Phi_{i,j} = J_0\left((i-j)\cdot2\pi d/ \lambda\right).
\end{align}
Here, $J_0(\cdot)$ is the zeroth-order Bessel function of the first kind. Also, $\lambda$ represents the wavelength which is equal to $c/f_\text{c}$ where $c$ is the speed of light and $f_\text{c}$ is the carrier frequency. 

In our model, we define $d$ as the effective distance between the position of the PA at time $t$ and the position of the RA at time $t+\delta$, as can be seen in Fig. \ref{system}. That is,  $d$ is given  by
\begin{align}\label{eq_d}
    d = |d_\text{a} - d_\text{m} | = |d_\text{a} - v\delta|,
\end{align}
where $d_\text{m}$ is the moving distance between  $t$ and $t+\delta $ while $v$ is the velocity of the vehicle. To conclude, different values of $v$, $ \delta$, $f_\text{c}$ and $d_\text{a}$ in (\ref{eq_d}) correspond to different values of $d$, which leads to different levels of channel spatial correlation (\ref{eq_tildeH})-(\ref{eq_phi}).

Combining (\ref{eq_tildeH}) and (\ref{eq_phi}) with normalization,  we have
\begin{align}\label{eq_H}
    H = \sqrt{1-\sigma^2} \hat{H} + \sigma q,
\end{align}
where $q \sim \mathcal{CN}(0,1)$ which is independent of the known channel value $\hat{H}\sim \mathcal{CN}(0,1)$, and $\sigma$ is a function of the mis-matching distance $d$.

From (\ref{eq_H}), for a given $\hat{H}$ and $\sigma \neq 0$, $|H|$ follows a Rician distribution, i.e., the probability density function (PDF) of $|H|$ is given by $f_{|H|\big|\hat{H}}(x) = \frac{2x}{\sigma^2}e^{-\frac{x^2+\hat{g}}{\sigma^2}}I_0\left(\frac{2x\sqrt{\hat{g}}}{\sigma^2}\right)$, where $ \hat{g} = (1-\sigma^2) |\hat{H}|^2$ Then, we define the channel gain between BS-RA as $ g = |{H}|^2$. By changing variables from $H$ to $g$, the PDF of $f_{G|\hat{H}}$ is given by
\begin{align}\label{eq_pdf}
    f_{G|\hat{H}}(x) = \frac{1}{\sigma^2}e^{-\frac{x+\hat{g}}{\sigma^2}}I_0\left(\frac{2\sqrt{x\hat{g}}}{\sigma^2}\right),
\end{align}
which is non-central Chi-squared distributed, and the cumulative distribution function (CDF) is
\begin{align}\label{eq_cdf}
    F_{G|\hat{H}}(x) = 1 - \mathcal{Q_M}\left( \sqrt{\frac{2\hat{g}}{\sigma^2}}, \sqrt{\frac{2x}{\sigma^2}}  \right).
\end{align}
Here, $I_n(x) = (\frac{x}{2})^n \sum_{i=0}^{\infty}\frac{(\frac{x}{2})^{2i} }{i!\Gamma(n+i+1)}$ denotes the $n$-th order modified Bessel function of the first kind and $ \mathcal{Q_M}(s,\rho) = \int_{\rho}^{\infty} xe^{-\frac{x^2+s^2}{2}}I_0(sx)\,\text{d}x$ is the Marcum $Q$-function \cite{Bocus2013CLapproximation}.

\vspace{-3mm}
\section{Analytical Results}
We assume that $d_\text{a}$, $\delta $ and $\hat{g}$ are known at the BS. It can be seen from (\ref{eq_pdf}) that $f_{G|\hat{H}}(x)$ is a function of $v$. For a given $v$, the distribution of $g$ is known at the BS, and a rate adaption scheme can be performed to improve the system performance.

The data is transmitted with rate $R^*$ nats-per-channel-use (npcu). If the gain instantaneous realization supports the data rate, i.e., $\log(1+gP)\ge R^*$, the data can be successfully decoded, otherwise outage occurs.  Hence, the outage probability in each time slot is obtained as $\text{Pr}(\text{outage}|\hat{H}) = F_{G|\hat{H}}\left(\frac{e^{R^*}-1}{P}\right)$. Considering slotted communication in block fading channels, where $\Pr(\text{Outage})>0$ depending on the fading model, throughput defined as the data rate times the successful decoding probability, i.e., the expected data rate successfully received by the receiver, is an appropriate performance metric \cite[p. 2631]{Biglieri1998TITfading}\cite[Th. 6]{Verdu1994TITgeneral}\cite[eq. (9)]{Makki2014TCperformance}. Hence, the rate adaptation problem maximizing the outage-constrained throughput in each time slot, with given $v$ and $\hat{g}$, can be expressed as
\begin{align}\label{eq_avgR}
    R_{\text{opt}|\hat{g}}=\argmax_{R^*\geq 0} \left\{ \left(1-\text{Pr}\left(\log(1+gP)<R^*\right)\right)R^* \right\}\nonumber\\
   =\argmax_{R^*\geq 0} \left\{ \left(1-\text{Pr}\left(g<\frac{e^{R^*}-1}{P}\right)\right)R^*\right \}\nonumber\\
    =\argmax_{R^*\geq 0} \left\{ \left(1-F_{G|\hat{H}}\left(\frac{e^{R^*}-1}{P}\right)\right) R^*\right\},
\end{align}
and the expected throughput is obtained by $\mathbb{E}\left\{ R_{\text{opt}|\hat g}\left(1-F_{G|\hat{H}}\left(\frac{e^{R_{\text{opt}|\hat g}}-1}{P}\right)\right) \right\}$ with expectation over $\hat g$. 

Using (\ref{eq_cdf}), (\ref{eq_avgR}) is simplified as
\begin{eqnarray}\label{eq_optR}
    R_{\text{opt}|\hat{g}}=\argmax_{R^*\geq 0} \left\{ \mathcal{Q_M}\left(  \sqrt{\frac{2\hat{g}}{\sigma^2}}, \sqrt{\frac{2(e^{R^*}-1)}{P\sigma^2}} \right)R^* \right\}.
\end{eqnarray}

Equation (\ref{eq_optR}) does not have a closed-form solution. For this reason, Lemma 1 derives an approximation for the optimal data rate maximizing the instantaneous throughput.

\begin{lem}
For a given channel realization $\hat{g}$, the throughput-optimized rate allocation is given by (\ref{eq_appRF}).
\end{lem}
 
\begin{proof}
We use the approximation of the first-order Marcum $Q$-function \cite[eq. (2), (7)]{Bocus2013CLapproximation}
\begin{align}
    \mathcal{Q_M} (s, \rho) &\simeq e^{\left(-e^{\mathcal{I}(s)}\rho^{\mathcal{J}(s)}\right)}, \nonumber\\
    \mathcal{I}(s)& = -0.840+0.327s-0.740s^2+0.083s^3-0.004s^4,\nonumber\\
    \mathcal{J}(s)& = 2.174-0.592s+0.593s^2-0.092s^3+0.005s^4,
\end{align}

which is accurate at low/moderate values of $s$, to obtain the optimal rate allocation (\ref{eq_optR}) by setting
\begin{align}
    \omega(\sigma, \hat{g}) &= e^{\left(\mathcal{I}\left(\sqrt{\frac{2\hat{g}}{\sigma^2}}\right)\right)}\left(\frac{2}{P\sigma^2}\right)^{\frac{\mathcal{J}\left(\sqrt{\frac{2\hat{g}}{\sigma^2}}\right)}{2}}, \\
    \nu(\sigma, \hat{g}) &= \frac{\mathcal{J}\left(\sqrt{\frac{2\hat{g}}{\sigma^2}}\right)}{2}.
\end{align}
In this way, (\ref{eq_optR}) is approximated as
\begin{align}\label{eq_appR}
    R_{\text{opt}|\hat{g}} \simeq \argmax_{R^*\geq 0} \left\{R^*e^{\left(-\omega(\sigma, \hat{g})\left(e^{R^*}-1\right)^{\nu(\sigma, \hat{g})}\right)}\right\}.
\end{align}
Then, setting the derivative of (\ref{eq_appR}) equal to zero, we obtain the throughput-optimized instantaneous rate as
\begin{align}\label{eq_appRF}
    R_{\text{opt}|\hat{g}} & = \operatorname*{arg}_{R^*\geq 0} \bigg\{ e^{-\omega\left(e^{R^*}-1\right)^{\nu}}\nonumber\\
    & \times\left(1-R^*\omega\nu\left(e^{R^*}-1\right)^{(\nu-1)}e^{R^*}\right)=0\bigg\}\nonumber\\
    & = \operatorname*{arg}_{R^*\geq 0} \left\{ R^*\omega\nu e^{R^*}\left(e^{R^*}-1\right)^{\nu-1}=1 \right\}\nonumber\\
    & \overset{(a)}{\simeq} \operatorname*{arg}_{R^*\geq 0}\left\{R^*\nu e^{R^*\nu}=\frac{1}{\omega}\right\}\nonumber\\
    & \overset{(b)}{=} \frac{1}{\nu(\sigma, \hat{g})}\mathcal{W}\left(\frac{1}{\omega(\sigma, \hat{g})}\right).
\end{align}
Here, $(a)$ comes from $e^{R^*}-1 \simeq e^{R^*} $ which is appropriate at moderate/high values of $R^*$. Also, $(b)$ is obtained by using the Lambert $\mathcal{W}$-function $xe^x = y \Leftrightarrow x = \mathcal{W}(y)$ \cite{corless1996lambertw}. 

\end{proof}

\vspace{-10mm}
\subsection{Temporal Correlation}
Here, we study the system performance in temporally-correlated conditions. Particularly, using the same model as in  \cite[eq. (2)]{Makki2013TCfeedback}, we further develop our channel model  (\ref{eq_H}) as
\begin{align}\label{eq_Htp}
    H_{k+1} = \beta H_{k} + \sqrt{1-\beta^2} z, 
\end{align}
for each time slot $k$, where $z \sim \mathcal{CN}(0,1)$ is a Gaussian noise which is uncorrelated with $H_{k}$. Also, $\beta$ is a known correlation factor which represents two successive channel realizations dependencies by $\beta = \frac{\mathbb{E}\{H_{k+1}H_{k}^*\}}{\mathbb{E}\{|H_k|^2\}}$. Substituting (\ref{eq_Htp}) into (\ref{eq_H}), we have
\begin{align}\label{eq_Ht}
    H_{k+1} = \beta\sqrt{1-\sigma^2}\hat{H}_{k}+\beta\sigma q+\sqrt{1-\beta^2}z.
\end{align}
To simplify the calculation,  $\beta\sigma q + \sqrt{1-\beta^2}z$ is equivalent to a new Gaussian variable $w \sim\mathcal{CN}\left(0,(\beta\sigma)^2+1-\beta^2\right)$. We can follow the same procedure as in (\ref{eq_avgR})-(\ref{eq_appRF}) to analyze the system performance in temporally-correlated conditions. Moreover, we can follow the same approach as in \cite{Wang2007TWCperformance} to add the effect of estimation error of $\hat{H}$ as an independent additive Gaussian variable whose variance is given by the accuracy of CSI estimation. In this way, we can follow the same approach as (\ref{eq_Htp})-(\ref{eq_Ht}) to study the effect of imperfect $\hat{H}$. However, due to space limits, we do not add the detailed analysis here.

\vspace{-3mm}
\section{{Simulation Results}}
In this section, we present the simulation results for the throughput as well as the outage probability in different cases with both spatial and temporal correlations:
\begin{figure}
\vspace{-3mm}
\centering
  \includegraphics[width=0.73\columnwidth]{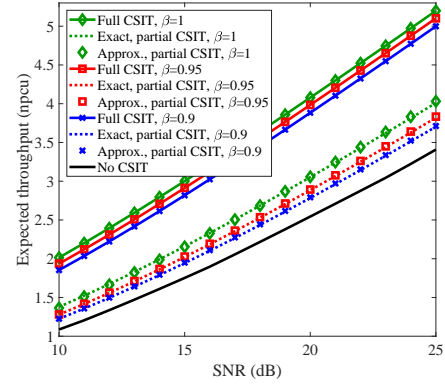}\\
 \vspace{-1mm}
\caption{Expected throughput in different cases, $v$ = 117 km/h.}\label{fig_SNRvsR}
\vspace{-3mm}
\end{figure}

\begin{figure}
\vspace{-1mm}
\centering
  \includegraphics[width=0.73\columnwidth]{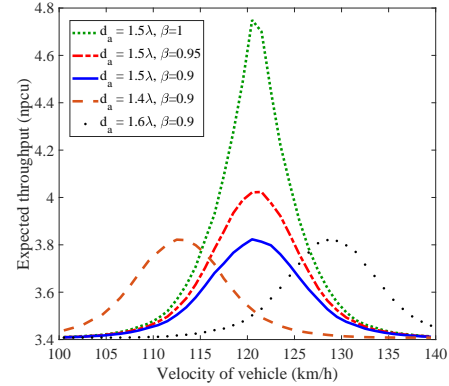}\\
\caption{Expected throughput for different velocities with SNR = 25 dB, in the case of partial CSIT, exact.}\label{fig_VvsR}
\vspace{-5mm}
\end{figure}
\begin{itemize}
    \item \textbf{Full CSIT}:  This ideal case assumes that the BS has perfect CSIT at each position without uncertainty/outage probability, i.e., $H = \beta\sqrt{1-\sigma^2}\hat{H}$, which gives an upper bound of the system performance.
    \item \textbf{Partial CSIT, exact}: This case presents the performance of (\ref{eq_avgR}), where we use the PA to obtain the partial CSIT and perform the rate adaptation. The uncertainties are from  both the spatial mismatching as well as the temporal correlation of the channel.
    \item \textbf{Partial CSIT, approximation}: This case presents the approximation results obtained by (\ref{eq_appRF}).
    \item \textbf{No CSIT}: $\sigma = 1$ in (\ref{eq_H}), i.e., $H \in \mathcal{CN}(0,1)$. In this case no PA is applied and it gives a lower bound of the system performance.
\end{itemize}

In the simulations, we set $ \delta$ = 5 ms, $f_\text{c}$ = 2.68 GHz.  We fix $d_\text{a} = 1.5\lambda$ except for Figs. \ref{fig_VvsR} and \ref{fig_VvsO}. Each point in the figures is obtained by averaging the system performance over $1\times10^5$ channel realizations of $\hat{H}$. In Fig. \ref{fig_SNRvsR}, we plot the expected throughput  in different cases for a broad range of signal-to-noise ratios (SNRs) which, because the noise has unit variance, we define as $10\log_{10}P$. Here, we set $v = 117$ km/h as used in (\ref{eq_d}). The analytical results obtained by Lemma 1, i.e., the approximation of (\ref{eq_avgR}), are also presented for $\beta = 1, 0.95, 0.9$. Also, the figure shows the results with no CSIT, which is independent of temporal correlation, i.e., $\beta$. Fig. \ref{fig_VvsR} evaluates the  expected throughput for a broad range of velocities with different values of $\beta$. Here, the SNR is set to 25 dB. Figs. \ref{fig_SNRvsO} and \ref{fig_VvsO} demonstrate the outage probability versus the SNR as well as the velocity for different values of $\beta$. In Fig. \ref{fig_SNRvsO}, the speed of the vehicle is set to 117 km/h while the SNR in Fig. \ref{fig_VvsO} is set to 10 dB. 

\begin{figure}
\centering
  \includegraphics[width=0.73\columnwidth]{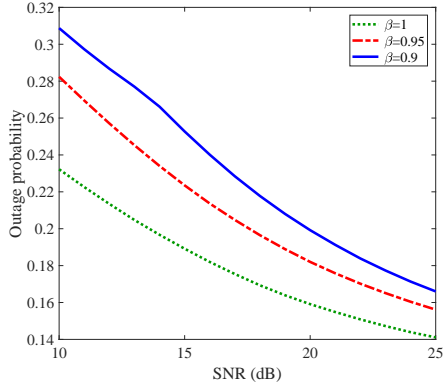}\\
\caption{Outage probability for a broad range of SNRs. $v$ = 117 km/h, in the case of partial CSIT, exact.}\label{fig_SNRvsO}
\vspace{-5mm}
\end{figure}

\begin{figure}
\centering
  \includegraphics[width=0.73\columnwidth]{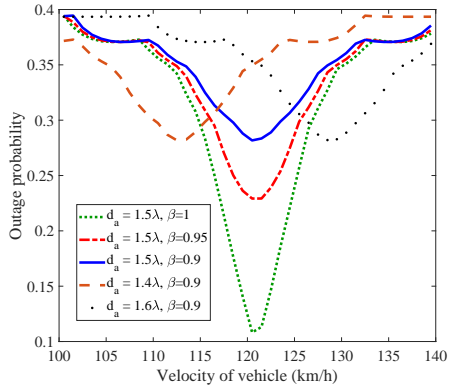}\\
\caption{Outage probability for different velocities with SNR=10 dB, in the case of partial CSIT, exact.}\label{fig_VvsO}
\vspace{-6mm}
\end{figure}

According to the figures, we can conclude the following:
\begin{itemize}
    \item The approximation scheme of Lemma 1 is very tight for a broad range of values of SNR (Fig. \ref{fig_SNRvsR}). Thus, for different parameter settings, the throughput-optimized rate allocation can be well determined by (\ref{eq_appRF}).
    \item With deployment of the PA, remarkable throughput gain is achieved especially in high SNRs. Also, the throughput increases when the temporal correlation is high, i.e., $\beta$ increases. Finally, as can be seen in Fig. \ref{fig_SNRvsO}, the outage probability decreases with SNR as well as $\beta$  while the sensitivity to temporal correlation increases as $\beta$ decreases.
    \item In Figs. \ref{fig_VvsR} and \ref{fig_VvsO},  we study the throughput as well as the outage probability for different speeds $v$. Here, manipulating $v$ is actually equivalent to changing the level of spatial correlation for given values of $ \delta$, $f_\text{c}$ and $d_\text{a}$ in (\ref{eq_d}). Also, we address the effect of different values of $d_\text{a}$ on the system performance. With the setup of Figs. \ref{fig_VvsR} and \ref{fig_VvsO}, the optimal speed, in terms of throughput and outage probability, is almost 120 km/h. This matches the straightforward calculation from (\ref{eq_d}) where the optimal speed without mismatching is 120.9 km/h. Moreover, when $\beta$ decreases, the throughput in Fig. \ref{fig_VvsR} as well as the outage probability in Fig. \ref{fig_VvsO} become less sensitive to the speed. To address the effect of $d_\text{a}$ on the system performance, we present the throughput as well as the outage probability for $d_\text{a}=1.4\lambda$ and $d_\text{a}=1.6\lambda$ in Figs. 3 and 5, respectively. As can be seen, the optimal velocity, in terms of the throughput and the outage probability, increases with $d_\text{a}$.
\end{itemize}

\vspace{-2mm}
\section{Conclusion}
We studied the impact of the mismatching problem in PA systems and developed rate adaptation methods to compensate for its effect. The simulation and analytical results show that, while PA-assisted adaptive rate adaptation leads to considerable performance improvement, the throughput and the outage probability are remarkably affected by the spatial mismatch and degree of temporal correlation.

\vspace{-3mm}
\bibliographystyle{IEEEtran}

\bibliography{Hao_WCL2019-1035.bib}

\begin{thebibliography}{10}
\providecommand{\url}[1]{#1}
\csname url@samestyle\endcsname
\providecommand{\newblock}{\relax}
\providecommand{\bibinfo}[2]{#2}
\providecommand{\BIBentrySTDinterwordspacing}{\spaceskip=0pt\relax}
\providecommand{\BIBentryALTinterwordstretchfactor}{4}
\providecommand{\BIBentryALTinterwordspacing}{\spaceskip=\fontdimen2\font plus
\BIBentryALTinterwordstretchfactor\fontdimen3\font minus
  \fontdimen4\font\relax}
\providecommand{\BIBforeignlanguage}[2]{{%
\expandafter\ifx\csname l@#1\endcsname\relax
\typeout{** WARNING: IEEEtran.bst: No hyphenation pattern has been}%
\typeout{** loaded for the language `#1'. Using the pattern for}%
\typeout{** the default language instead.}%
\else
\language=\csname l@#1\endcsname
\fi
#2}}
\providecommand{\BIBdecl}{\relax}
\BIBdecl

\bibitem{Yutao2013ICMmoving}
S.~Yutao, J.~Vihriala, A.~Papadogiannis, M.~Sternad, W.~Yang, and T.~Svensson,
  ``Moving cells: a promising solution to boost performance for vehicular
  users,'' \emph{IEEE Commun. Mag.}, vol.~51, no.~6, pp. 62--68, Jun. 2013.

\bibitem{Sternad2012WCNCWusing}
M.~Sternad, M.~Grieger, R.~Apelfr\"ojd, T.~Svensson, D.~Aronsson, and A.~B.
  Martinez, ``Using predictor antennas for long-range prediction of fast fading
  for moving relays,'' in \emph{Proc. IEEE Wireless Commun. Netw. Conf.
  Workshops (WCNCW)}, Paris, France, Apr. 2012, pp. 253--257.

\bibitem{BJ2017ICCWusing}
J.~Bj\"orsell, M.~Sternad, and M.~Grieger, ``Using predictor antennas for the
  prediction of small-scale fading provides an order-of-magnitude improvement
  of prediction horizons,'' in \emph{Proc. IEEE Int. Conf. Commun. Workshops
  (ICC Workshops)}, Paris, France, May 2017, pp. 54--60.

\bibitem{phan2018WSAadaptive}
D.-T. Phan-Huy, S.~Wesemann, J.~Bj\"orsell, and M.~Sternad, ``{Adaptive massive
  MIMO for fast moving connected vehicles: It will work with predictor
  antennas!}'' in \emph{Proc. 22nd Int. ITG Workshop Smart Antennas (WSA)},
  Bochum, Germany, Mar. 2018, pp. 1--8.

\bibitem{DT2015ITSMmaking}
D.-T. Phan-Huy, M.~Sternad, and T.~Svensson, ``{Making 5G adaptive antennas
  work for very fast moving vehicles},'' \emph{IEEE Intell. Transp. Syst.
  Mag.}, vol.~7, no.~2, pp. 71--84, Apr. 2015.

\bibitem{Shin2003TITcapacity}
H.~Shin and J.~H. Lee, ``Capacity of multiple-antenna fading channels: spatial
  fading correlation, double scattering, and keyhole,'' \emph{IEEE Trans. Inf.
  Theory}, vol.~49, no.~10, pp. 2636--2647, Oct. 2003.

\bibitem{Chizhik2000CLeffect}
D.~Chizhik, F.~Rashid-Farrokhi, J.~Ling, and A.~Lozano, ``Effect of antenna
  separation on the capacity of {BLAST} in correlated channels,'' \emph{IEEE
  Commun. Lett.}, vol.~4, no.~11, pp. 337--339, Nov. 2000.

\bibitem{Bocus2013CLapproximation}
M.~Z. {Bocus}, C.~P. {Dettmann}, and J.~P. {Coon}, ``{An approximation of the
  first order Marcum $Q$-function with application to network connectivity
  analysis},'' \emph{IEEE Commun. Lett.}, vol.~17, no.~3, pp. 499--502, Mar.
  2013.

\bibitem{Biglieri1998TITfading}
E.~{Biglieri}, J.~{Proakis}, and S.~{Shamai}, ``Fading channels:
  information-theoretic and communications aspects,'' \emph{IEEE Trans. Inf.
  Theory}, vol.~44, no.~6, pp. 2619--2692, Oct. 1998.

\bibitem{Verdu1994TITgeneral}
S.~{Verd\'{u}} and {Te Sun Han}, ``A general formula for channel capacity,''
  \emph{IEEE Trans. Inf. Theory}, vol.~40, no.~4, pp. 1147--1157, Jul. 1994.

\bibitem{Makki2014TCperformance}
B.~{Makki} and T.~{Eriksson}, ``{On the performance of MIMO-ARQ systems with
  channel state information at the receiver},'' \emph{IEEE Trans. Commun.},
  vol.~62, no.~5, pp. 1588--1603, May 2014.

\bibitem{corless1996lambertw}
R.~M. Corless, G.~H. Gonnet, D.~E. Hare, D.~J. Jeffrey, and D.~E. Knuth, ``{On
  the Lambert W function},'' \emph{Adv. Comput. Math.}, vol.~5, no.~1, pp.
  329--359, 1996.

\bibitem{Makki2013TCfeedback}
B.~Makki and T.~Eriksson, ``{Feedback subsampling in temporally-correlated
  slowly-fading channels using quantized CSI},'' \emph{IEEE Trans. Commun.},
  vol.~61, no.~6, pp. 2282--2294, Jun. 2013.

\bibitem{Wang2007TWCperformance}
C.~{Wang}, E.~K.~S. {Au}, R.~D. {Murch}, W.~H. {Mow}, R.~S. {Cheng}, and
  V.~{Lau}, ``{On the performance of the MIMO zero-forcing receiver in the
  presence of channel estimation error},'' \emph{IEEE Trans. Wireless Commun.},
  vol.~6, no.~3, pp. 805--810, Mar. 2007.

\end{thebibliography}

\end{document}